\documentclass[acmsmall,screen,nonacm]{acmart}
\usepackage[utf8]{inputenc}
\usepackage{amsmath}
\usepackage{enumerate}
\usepackage{mathpartir}
\usepackage{mathtools}
\usepackage{pmboxdraw}
\usepackage[T1]{fontenc}

\usepackage[normalem]{ulem}
\usepackage{mathpartir}
\usepackage{proof}
\usepackage{mathtools}
\usepackage{longtable}
\usepackage{supertabular}
\usepackage{ stmaryrd }

\usepackage{comment}
\usepackage{color}   
\hypersetup{
    colorlinks,
    citecolor=black,
    filecolor=black,
    linkcolor={blue!50!black},
    urlcolor=black
}

\usepackage{tikz} 

\usepackage{diagbox}
\usepackage{pdflscape}

\usepackage{makecell}
\usepackage{changepage}
\usepackage{xcolor}
\usepackage{multicol}
\usepackage{multirow}
\usepackage{listings}
\usepackage{etoolbox}
\usepackage{environ}
\usepackage{newunicodechar}
\usepackage{bm}
\usepackage{vwcol}  

\usepackage{framed}
\usepackage{cprotect}

\newunicodechar{∞}{\ensuremath{\mathnormal\infty}}
\newunicodechar{♯}{\ensuremath{\mathnormal\sharp}}
\newunicodechar{⌊}{\ensuremath{\mathnormal{\lfloor}}}
\newunicodechar{⌋}{\ensuremath{\mathnormal{\rfloor}}}
\newunicodechar{≟}{\ensuremath{\mathnormal{\stackrel{?}{=}}}}
\newunicodechar{≡}{\ensuremath{\mathnormal{\equiv}}}

\newcommand\validsig{\mathrel{\operatorname{\mathsf{sig}}}}
\newcommand\validctx{\mathrel{\operatorname{\mathsf{ctx}}}}
\newcommand\cotype{\operatorname{\mathsf{cotype}}}
\newcommand\type{\operatorname{\mathsf{type}}}
\newcommand\typecotype{\operatorname{\mathsf{(co)type}}}
\newcommand\kind{\operatorname{\mathsf{kind}}}

\newcommand{\m}[1]{\operatorname{\mathsf{{#1}}}}
\newcommand{\mi}[1]{\mathop{\operatorname{\mathit{{#1}}}}}

\newcommand{\mt}[1]{\mathop{\operatorname{\mathtt{{#1}}}}}

\newcommand{\pitype}[3]{\ensuremath{\Pi {#1} : {#2} .\, {#3}}}

\newcommand\hra\hookrightarrow

\newcommand{\pdepth}[1]{_{({#1})}}
\newcommand{\depth}[1]{\pdepth{#1}}
\newcommand{\dk}{\depth{k}}
\newcommand{\dkpo}{\depth{k+1}}

\newcommand\erased[1]{{{#1}^o}}

\newcommand{\enc}[1]{\ulcorner {#1}\urcorner}

\newcounter{numbered}
\newenvironment{numbered}{
    \setcounter{numbered}{0}
    }{
    \setcounter{numbered}{0}
    }

\newcommand{\previousNumber}{\the\numexpr\value{numbered}-1\relax }

\newcommand{\snil}{\operatorname{\mathsf{()}}}

\newtoggle{journalversion}
\toggletrue{journalversion}

\iftoggle{journalversion}{
}{
}

\iftoggle{journalversion}{
    \includeonly{}
}{
    \includeonly{appendix}
}

\acmConference[Conference acronym 'XX]{Make sure to enter the correct
  conference title from your rights confirmation emai}{June 03--05,
  2018}{Woodstock, NY}
\acmPrice{15.00}
\acmISBN{978-1-4503-XXXX-X/18/06}

\author{Zhibo Chen}
\email{zhiboc@andrew.cmu.edu}
\orcid{0000-0003-0045-5024}
\affiliation{%
  \institution{Carnegie Mellon University}
  \country{USA}
}

\title{A Logical Framework with Infinitary Terms}

\begin{document}

\begin{abstract}
Logical frameworks are successful in modeling proof systems.
Recently, CoLF extended the logical framework LF to 
support higher-order rational terms that enable adequate encoding of circular objects and derivations. 
In this paper, we propose CoLF$^\omega$ as
an alternative interpretation of CoLF-style signatures where terms are taken to be 
all possibly infinitary terms that are consistent with a given signature. In particular, 
we propose the notion of productive B\"ohm trees, 
a particular kind of typed $\bot$-free B\"ohm trees that are 
closed under hereditary substitution. 
We show that the productive B\"ohm trees are capable of meta-encoding their own structure.
Overall, we hope to establish CoLF$^\omega$ as a new formal framework for the encoding of 
infinitary regular and non-regular structures.

\end{abstract}
\maketitle

    
\section{Introduction}

Infinite objects are representable in the logical framework LF by indexing a type family with 
a natural number as its observation depth. For example, in the following signature, 
the stream of natural numbers whose first $k$ elements can be observed 
is in compositional bijection with the canonical terms of the type family $\mt{stream}\, (\mt{succ}^{k}\, \mt{zero})$.
\begin{verbatim}
nat : type.
zero : nat.
succ : nat -> nat.
stream : nat -> type.
unobservable : stream zero.
cocons : {k : nat} nat -> stream k -> stream (succ k).
\end{verbatim}

The encoding is hard to work with, because the observation depth of the stream needs to be tracked everywhere a stream is used.
CoLF \cite{Chen23fossacs} is an extension of the logical framework LF that supports natural and adequate encodings of 
circular objects and circular derivations. To 
make type checking decidable,
CoLF limits its term model to higher-order rational terms. This limitation has the 
shortcoming that
objects without a regular structure cannot be adequately represented in CoLF.
For example, the stream of natural numbers with repeating $1$'s and $2$'s,
$1, 2, 1, 2, \dots$, can be encoded in CoLF because it has a regular structure, i.e. 
the stream can be given by the equation $S = 1,2, S$. The stream of natural numbers 
counting up from $1$, $1, 2, 3, 4, \dots$, cannot be encoded in CoLF because it does not have a regular structure, 
i.e. the stream cannot be given by a system of equations.
In this paper, we develop a new type theory CoLF$^\omega$, which provides an alternative term model 
for CoLF-style signatures where terms are taken to be all possibly infinitary terms.
Many more interesting infinitary objects can be encoded in CoLF$^\omega$ than in CoLF.



The main contributions of this paper are:
\begin{itemize}
    \item 
    A formulation of infinitary syntax trees (Section~\ref{sec:infinitary_syntax_trees}).
    \item
    A definition of productive B\"ohm trees via the infinitary syntax trees (Section~\ref{sec:productive_bohm_trees}).
    \item 
    The type theory of CoLF$^\omega$, whose terms are productive B\"ohm trees (Section~\ref{sec:the_type_theory_of_colf_omega}). 
    \item 
    An interpretation of (adapted) finitary signatures of CoLF into CoLF$^\omega$ (Section~\ref{sec:interpretation_of_colf_signatures}).
    \item
    A meta-encoding of the productive B\"ohm trees using CoLF$^\omega$ signatures (Section~\ref{sec:encoding_productive_bohm_trees}).
    \item 
    A case study on co-natural numbers and co-binary numbers using CoLF$^\omega$ (Section~\ref{sec:conatural_numbers_and_cobinary_numbers}).
\end{itemize}


\section{Examples of CoLF$^\omega$}
\label{sec:examples_of_colf_omega}

We illustrate informally the infinitary term model of  CoLF$^\omega$, and how 
it is different from the rational term model of CoLF.

\subsection{Streams}

Consider the following CoLF signature for defining streams of natural numbers.

\begin{verbatim}
nat : type.
zero : nat.
succ : nat -> nat.
stream : cotype.
cocons : nat -> stream -> stream.
\end{verbatim}

In CoLF, the only terms are \emph{rational terms}, i.e. terms having finitely many subterms up to 
renaming of free and bound variables. Canonical terms of type \verb$stream$ are rational.
As a consequence, we can only represent rational streams (streams with finitely many distinct repeating patterns) in the framework.
A stream that counts up from a certain natural number or a stream that enumerates all 
Fibonacci numbers is not a term of type \verb$stream$ in CoLF.
However, in CoLF$^\omega$, 
all streams are infinitary terms consistent with the signature. 
That is,  the canonical terms of type \verb$stream$ include all possible streams, and there are uncountably many of them.

 There is a question of 
whether noncomputable streams are also represented in the canonical terms.
For example, temperature readings from a measurement device can be a stream of 
natural numbers that is not computable. 
We choose to leave open the question of computability intentionally in the hope that 
the development of CoLF$^\omega$ can be used to encode either computable or noncomputable objects, 
as long as the choice is made consistently.

While the canonical terms of type \verb$stream$ can be any stream, we may specify the streams we actually care 
about using predicates.
For instance, we can specify a stream that counts up from 
the natural number \verb$n$ by saying that the stream \verb$S$ is a term of type \verb$stream$ such that 
\verb$up n S$ holds, where \verb$up$ is the predicate defined below.

\begin{verbatim}
up : nat -> stream -> cotype.
up/def : {N : nat} {S : stream} up (succ N) S
    -> up N (cocons N S).
\end{verbatim}

A term of type \verb$up N S$ must be a term of the form \verb$up/def N (cocons N S') U$, where \verb$S = cocons N S'$, and \verb$U$ 
is a term of type \verb$up (succ N) S'$.
In fact, there is a unique inhabitant of type \verb$up N S$ for each $N$.
We show an example of a term of type \verb$up zero S$ where \verb$S$ is required to be a stream that counts up from $0$.
To reduce visual clutter, we write \verb$N :: S$ for \verb$cocons N S$, \verb$0$ for \verb$zero$, \verb$1$ for \verb$succ zero$, etc.
We have
\begin{verbatim}
up/def 0 S : up 1 S -> up 0 (0 :: S).
up/def 1 S : up 2 S -> up 1 (1 :: S).
up/def 2 S : up 3 S -> up 2 (2 :: S).
...
\end{verbatim}
Then,
\begin{verbatim}
m : up 0 (0 :: 1 :: 2 :: 3 :: ... ) = 
    up/def 0 (1 :: 2 :: 3 :: ...)
        (up/def 1 (2 :: 3 :: ...)
            (up/def 2 (3 :: ...)
                (up/def 3 ... ...
                )
            )
        ).
\end{verbatim}

The term $\mt{m}$ is not typeable in CoLF exactly because it is not a rational term: 
the set of its subterms contains $\mt{succ}^n \, \mt{zero}$ for every natural number $n$.

\subsection{Fibonacci Sequences}

We could define complex infinitary streams where 
the later parts of the stream depend on earlier parts.
Let \verb$fib n m S$ denote the stream of the Fibonacci sequence whose previous two 
numbers are $n$ and $m$. The whole Fibonacci stream starting with $1$ would be the stream $S$ 
such that the type \verb$fib zero (succ zero) S$ is inhabited.

\begin{verbatim}
fib : nat -> nat -> stream -> cotype.
fib/def : add X Y Z
    -> fib Y Z S
    -> fib X Y (cocons Z S).
\end{verbatim}

The \verb$add X Y Z$ predicate is defined inductively and is inhabited if $\mt{X + Y = Z}$.

\subsection{Real Numbers}
\label{sec:real_numbers}

A bit stream $b_0, b_1, \dots$ can represent the binary expansion of a real number in $[0,1]$
\[b_0, b_1, \dots \rightsquigarrow \Sigma_{i=0}^{\infty}(b_i \cdot \frac{1}{2^{i+1}})\]
\begin{verbatim}
bitstream : cotype.
b0 : bitstream -> bitstream.
b1 : bitstream -> bitstream.
\end{verbatim}

For example, the real number $0.101010\dots$ (binary decimal expansion) 
can be represented by the bit stream $1, 0, 1, 0, 1, 0 \dots$ whose encoding in CoLF$^\omega$
is  \verb$b1 (b0 (b1 (b0 (b1 (b0 ...)))))$.
Note that this real number is actually a rational number and thus is representable 
in CoLF as the term 

\verb$n : bitstream = b1 (b0 n)$. 

An example of an irrational real number would be the number, $0.1010010001\dots$, and it can be 
represented in CoLF$^\omega$ as the following infinitary term.

\verb$m : bitstream = b1 (b0 (b1 (b0 (b0 (b1 (b0 (b0 (b0 (b1 ...)))))))))$.

This number is not representable in CoLF because the corresponding term is not a rational term.
It is easy to see that rational numbers correspond to rational terms, 
and irrational numbers correspond to irrational terms.
In summary,
CoLF$^\omega$ can represent all real numbers \footnote{
As with streams, we leave the issue of whether only computable reals 
are represented or all reals are represented as a decision that 
the user of CoLF$^\omega$ can make. The reader is referred to \citet{Bauer05cca}
for a discussion of the computability of real numbers.}
whereas 
CoLF can only represent rational ones between $0$ and $1$.

\section{Infinitary Syntax Trees}
\label{sec:infinitary_syntax_trees}

We give an account for infinitary syntax trees, which serve as the 
technical device for defining productive B\"ohm trees.

The concept of observation is central to infinitary structures. A finitary structure can be observed in its totality
with a single observation, whereas an infinitary structure cannot be observed in its totality with a single observation.
The number of remaining steps that a term may be observed is called the \emph{observation depth} and is
written using a number subscript in parentheses. 
Syntax categories will always have an observation depth attached. For example, 
we write $M\depth k$ for a term with observation depth $k$, and $A\depth k$ for a type with observation depth $k$.
When we write down the grammar for a possibly infinitary term, 
the \emph{infinitary grammar} will specify the cases when the term undergoes a single step of observation. 
Concretely, we specify the grammar for observation depth $k+1$ in terms of the grammar for terms of observation depth $k$.
That is, given a syntax category $M$, the grammar for $M\depth{k+1}$ is specified in terms of $M\depth{k}$ for coinductive definitions, 
and is specified in terms of $M\depth{k+1}$ for inductive definitions.
The grammar may be mutually recursive, in that the grammar for $M\depth{k+1}$ may be specified in terms of $M'\depth k$ or 
$M'\depth{k+1}$ where $M'$ is another syntax category.
We also assume a universal base case with observation depth $0$ for all syntax categories.
For example, we write $M\depth 0$ for an unobservable term and $A\depth0$ for an unobservable type.
The depth $\omega$ is used for non-finite depth. For example, 
we write $M\depth \omega$ for a term that can be observed indefinitely
and write $A \depth \omega$ for a type that can be observed indefinitely.
We may sometimes omit the depth $\omega$ annotation for a syntax category, e.g.
we may just write $M$ for $M\depth \omega$, and $A$ for $A \depth\omega$.


We illustrate our use of the infinitary grammar through a series of examples. The reader should be reminded 
that we are making infinitary structures directly and formally precise.

\begin{enumerate}
    \item 
Natural Numbers

Inductive grammars are used to define finitary structures. The inductive nature is 
exemplified by the fact that the grammar for all structures is defined at the current observation depth.

We have the definition for natural numbers $N$, transcribed from the usual inductive definition.

\[N\depth{k+1} ::= 0 \mid S\, N\depth{k+1}\]

The grammar specifies that in a single observation, a natural number $N$ is either $0$, or the successor of another natural number $N'$, 
where $N'$ must be observed in the same observation. Given the finite nature of the observation, 
a natural number is a series of $S$'s followed by $0$.

\item Conatural Numbers

A slightly modified grammar defines the conatural numbers $C$.
\[C\depth{k+1} ::= 0 \mid S\, C\depth{k}\]

The grammar specifies that in a single observation, a conatural number $C$ is either $0$ or the successor of 
another conatural number $C'$ where $C'$ must be observed later, because it has one less observation depth.

\item Bit streams

A stream of bits $\mi{BS}$ may be defined by the following grammar:
\[{\mi{BS}}\depth{k+1} = b0, {\mi{BS}}\depth{k} \mid b1, {\mi{BS}}\depth{k}\]

The grammar specifies that an observation of a bitstream is either the zero bit $b0$, or the one bit $b1$, 
followed by another bitstream that must be observed in subsequent steps.

\item 
Binary number 

A binary number $\mi{BN}$, is a bit stream of finite length, where the least significant bit is listed first. 
A binary number 
can be represented using the following grammar:
\[ {\mi{BN}}\depth{k+1} ::= b0, {\mi{BN}}\depth{k+1} \mid b1, {\mi{BN}}\depth{k+1} \mid \epsilon\]

The grammar specifies that an observation of a binary number will reveal that either it is empty, or a bit ($b0$ or $b1$) followed 
by another binary number that must be observed in the same observation. Because an observation may only reveal 
a finite amount of information, a binary number cannot have an infinite number of bits before $\epsilon$.

\item Finitely-padded streams

A finitely-padded stream \cite{Chen21ms,Chen23fossacs} (a.k.a. left-fair streams \cite{Basold18phd}) is a stream of natural numbers with a finite number of padding between any two numbers.
The grammar below specifies a coinductive finitely-padded stream $\mi{PS}$ and an inductive padding $\mi{P}$ defined recursively.

\[{\mi{PS}}\depth{k+1} ::= N\depth{k+1}, P\depth{k+1} \]
\[P\depth{k+1} ::= P\depth{k+1} \mid  {\mi{{PS}}}\depth{k}\]

The grammar specifies that a single observation on a finitely-padded stream will reveal it is a natural number, followed by a padding, both of which 
must be observed in the same observation. An observation on padding will reveal that it is either another padding, in which 
case this other padding must be further observed, or a finitely-padded stream, in which case the stream must be observed in the next observation.
Overall, an observation on a stream will reveal it is a natural number followed by a finite amount of padding, and then followed by a stream 
that must be observed in the next observation.

\item Different kinds of infinite $\lambda$-terms
 
\citet{Kennaway97}, and \citet{Barendregt09ic} observed three formulations of infinite lambda trees 
that have wide applications. They are B\"ohm trees ($BT$), L\'evy-Longo trees ($LLT$), and Berarducci trees ($BeT$). 
The essential difference is that B\"ohm trees may not contain infinite chains of applications or abstractions, 
L\'evy-Longo trees may contain infinite chains of abstractions but not infinite chains of applications, 
and Berarducci tress may contain both infinite chains of abstractions and applications.
All trees may not contain $\beta$-redexes.

Perhaps the easiest among the three is the grammar for Berarducci trees ($BeT$) as specified below. 
The grammar is broken into canonical terms $\mi{BeT}$ and neutral terms $\mi{BeT_{APP}}$. An observation
of a $\mi{BeT}$ tree will reveal that it is $\bot$, or an abstraction, whose subterm shall be observed in the next step, 
or a neutral term that must be observed in the same observation, 
while an observation of a neutral $\mi{BeT_{APP}}$ tree will reveal that it is either a head variable, or an application 
where each subterm must be observed in a subsequent observation.
\[{\mi{BeT}}\depth{k+1} ::= \bot \mid \lambda x. {\mi{BeT}}\depth{k} \mid ({\mi{BeT_{APP}}})\depth{k+1}\]
\[{\mi{BeT_{APP}}}\depth{k+1} ::=  x \mid (BeT_{APP})\depth{k}\, (BeT_2)\depth{k}\]

The grammar of L\'evy-Longo trees differs from Berarducci trees in that if the observation reveals an 
application, the observation must continue into the argument subterm, thereby disallowing infinite chains of applications.
\[{\mi{LLT}}\depth{k+1} ::= \bot \mid \lambda x. {\mi{LLT}}\depth{k} \mid ({\mi{LLT_{APP}}})\depth{k+1}\]
\[{\mi{LLT_{{APP}}}}\depth{k+1} ::=  x \mid (LLT_{APP})\depth{k+1}\, (LLT_2)\depth{k}\]

The grammar of B\"ohm trees has a further restriction that if the observation reveals $\lambda$-abstraction,
then the observation must continue into its body, thereby disallowing infinite chains of abstractions.
A single observation of a B\"ohm tree will reveal all abstractions, all applications, and finally the head variable.

\[{\mi{BT}}\depth{k+1} ::= \bot \mid \lambda x. {\mi{BT}}\depth{k+1} \mid ({\mi{BT_{APP}}})\depth{k+1}\]
\[{\mi{BT_{APP}}}\depth{k+1} ::=  x \mid (BT_{APP})\depth{k+1}\, (BT_2)\depth{k}\]

\end{enumerate}

In summary, infinitary syntax trees provide a formal foundation for infinitary structures, by stratifying 
an infinitary term into distinct chunks of observations. The distinct chunks are delineated through the 
concept of an observation depth.

\subsection{Equality}
\label{sec:equality_infinitary_syntax_tree}

We say that two potentially infinite syntax trees of observation depth $k$ are equal \textit{up to}  depth $k$, (notation $=\depth k$)
iff the observation of two terms up to depth $k$ does not reveal a difference between those two terms.
That is, $M\depth k = \depth k M' \depth k$ if the first $k$ observations of $M$ and $M'$ do not reveal a 
difference between them.

We always have the trivial case that $M\depth 0 = \depth 0 M' \depth 0$, that is, two terms are trivially equal because
the first zero observations of the two terms will not reveal a difference between them.
Given an infinitary grammar, the equality at depth $k+1$ can always be defined structurally. 
As an example, given the grammar for conat, 
\[C\depth{k+1} ::= 0 \mid S\, C\depth{k}\]
we define the equality by the following rules:
\begin{enumerate}
    \item (Trivially) 
    $C\depth 0 = \depth 0 C'\depth 0$

    \item $0 =\depth{k+1} 0$

    \item $S\, C\depth{k} =\depth{k+1} S\, C'\depth{k}$ if $C\depth{k} =\depth{k} C'\depth{k}$.

\end{enumerate}
Here, the first rule says that
    two unobservable terms are equal up to depth $0$.
    The second rule says that
    if an observation (on terms with depth $k+1$) reveals that both terms are zero, then they are equal up to depth $k+1$.
    The third rule says that
    if an observation on terms with depth $k+1$ reveals that the left-hand side is the successor followed by a term $C\depth k$ of depth $k$, 
    and the right-hand side is the successor followed by a term $C'\depth k$ of depth $k$, then 
    the terms of depth $k+1$ are equal up to depth $k+1$ if $C\dk$ and $C'\dk$ are equal up to depth $k$.

One may wonder if equality could be defined on terms with different observation depths, and we answer 
that because of the nature of observation, a term of any observation depth may be viewed as a term of a lesser observation depth 
by definition. That is, given a term $M\depth k$, we can construct a term $M\depth{j}$ with $j < k$ that mimics the behavior of $M\depth{k}$ 
for the first $j$ steps.
Therefore, the definition of equality on heterogeneous depths is not necessary.

As with the convention that we write $M\depth \omega$ or simply $M$ for terms of infinitary observation depth, 
we write $=\depth \omega$ or simply $=$ for equality relation on those infinitary terms.

\section{Productive B\"ohm Trees}
\label{sec:productive_bohm_trees}

The logical framework methodology establishes a bijective correspondence between the structures 
that we would like to encode and the terms of the logical framework.
In the case of LF logical framework, deductions are represented by dependently-typed $\lambda$-terms 
\cite{Harper93jacm,Harper07jfp}.
The dependently-typed $\lambda$-terms are just simply-typed $\lambda$-terms when 
the type annotation for $\lambda$-abstractions are erased \cite{Watkins02tr}.
The simply-typed $\lambda$-terms have two crucial properties that make it a suitable target for a logical framework.
First, every term has a $\beta$-normal-$\eta$-long form, which provides a basis for term equality modulo $\beta\eta$-conversion.
Second, the normal forms of the terms are closed under hereditary substitution, 
thereby enabling the higher-order encoding strategies. When infinitary structures become the target of the encoding, 
infinitary $\lambda$-terms become the natural choice for the term model of the logical framework.

None of the typed versions of the three kinds of infinitary $\lambda$-terms have our desired properties. 
First, they all contain the unsolvable term $\bot$, which has no place in the encoding of infinitary structures.
Even if the $\bot$ was removed from their structure, the term structures are not closed under hereditary substitution, 
(i.e. substitutions followed by $\beta\eta$-normalization). To see this, consider the 
term $F = \lambda x.\, x \, (x\,  (x \, (...)))$, and $F$ could be assigned 
the simple type $(* \to *) \to *$, where $*$ is a base type. Let the term $I$ denote the identity function, $I = \lambda z. \, z : * \to *$, 
we see that the term $F\, I$, or the substitution $[F/y](y \, I)$, is $I^\omega = I\, (I\, (I\, \dots))$.
This term does not normalize to a head normal form \cite{Barendregt09ic}.

We formulate the notion of typed productive B\"ohm trees as a subclass of B\"ohm trees, 
with constants and without $\bot$, that 
are closed under hereditary substitution. 
First, we add constants (or constructors) to B\"ohm trees by fixing an infinite set of variable names to serve as constant names. 
Those variables 
are subsequently called constants (syntax category $c$). When 
constructing a $\lambda$-abstraction, the binder name will never be one of the constant names, 
and constants never vary under substitution.
We also adopt the notion of head-spine form \cite{Watkins02tr} for iterative applications. For example, $x\, M_1\, M_2\, M_3$
is written $x\cdot (M_1; M_2; M_3)$. The infinitary grammar for productive B\"ohm trees is given below.

\begin{center}
    \begin{tabular}{lrcl}
        Canonical terms  &  $M\pdepth{k+1} $,$N\dkpo$ & $ ::= $ & $\lambda x.\, M\pdepth{k+1} \mid R\pdepth{k+1}$ \\
        Neutral terms   &  $R\pdepth{k+1} $ & $ ::= $ & $x \cdot T\pdepth{k+1} \mid  c\cdot S \pdepth{k+1}$\\
        Continuing Spines   & $T\pdepth{k+1}$ & $ ::= $ & $\snil \mid M\pdepth {k+1} ; T\pdepth{k+1}$ \\
        Suspended Spines  & $S\pdepth{k+1}$ & $ ::= $ & $\snil \mid M\pdepth {k} ; S\pdepth{k+1}$ \\
    \end{tabular}
\end{center}

The difference between productive ($\bot$-free) B\"ohm trees and non-productive $\bot$-free B\"ohm trees that 
in a single observation of productive B\"ohm trees, the arguments following a variable head must be observed in the 
same observation whereas in non-productive $\bot$-free B\"ohm trees, the arguments are always observed in a subsequent observation.
In other words, in productive B\"ohm trees, only when we encounter constants do we halt the current observation and defer the arguments to 
the next observation. The presence of constants gives rise to the notion of productivity, defined by a condition on the infinite traces.

A trace is a possibly infinite list of head variables or constants where each element is the head of a direct child 
of the preceding element.
Formally, the set of possibly infinite traces of $M\depth \omega$, $\m{traces}(M\depth\omega)$ is defined to be the following, where $h$ 
is either a variable or a constant.
\[
    \m{traces}(\lambda x_1. \dots \lambda x_l.\, h \cdot (M_1; \, \dots ;M_n)) =
    \begin{cases}
        \{h\} & \text{ if } n = 0\\
     \{h, T \mid T \in \cup_i\{\m{traces}(M_i)\}\} & \text{ if } n > 0
    \end{cases}
    \]

We say that a term is \textit{productive} if there are infinitely many occurrences of constants along each infinite trace of the term.
If we observe along any trace in a productive infinite term, we will always encounter a constructor in a finite number of steps. 
There can be only finitely many variables between any two 
constants along any infinite trace through a productive term $M$. 
For example, the term 
$Z = \lambda x. \, x \cdot (\lambda y. \, y \cdot (\lambda z. \, z \cdot (\dots)))$
is not productive as its only infinite trace is $x, y, z, \dots$ consisting of only variables. 
The term $F = \lambda x.\, x \cdot (x\cdot (x \cdot (...)))$ is not productive because 
its only trace, which is infinite, is $x, x, x, \dots$, and this trace consists of only variables.
The term $I = \lambda x. \, x$ is trivially productive because there are no infinite traces. 
The term $\lambda x. \, c \cdot (c \cdot (c \cdot (\dots)))$ is productive because 
the trace $c, c, c, \dots$ contains infinitely many occurrences of the constructors $c$. 
A similar reasoning shows that 
the term $\lambda x. \, c \cdot (x \cdot (c \cdot (x \cdot ( c \cdot (\dots)))))$ 
and the term $\lambda x. \, c \cdot (x \cdot (c \cdot (x \cdot (x \cdot (c \cdot (x \cdot (x \cdot (x \cdot (c \cdot (\dots))))))))))$ 
are both productive.

We show that the grammar for productive B\"ohm trees directly corresponds to the notion of productivity.
\begin{theorem} [Productivity]
    We have

    (1) Every canonical term is productive.

    (2) Every $\bot$-free productive B\"ohm tree is a canonical term.
\end{theorem}
\begin{proof}
    (1)  Given a canonical term  $M\depth\omega$, we show that there can only be finitely many variables between two constants on an infinite trace of $M$.
    For any $c \cdot S$ which is a subterm of $M$, an observation (which is always finitary) of $M$ will either involve $S$ or not.
    If it involves $S$, another constant has been encountered on this trace. If it does not involve $S$, then there is no infinitary trace because 
    observations are always finitary.

    (2) Given a B\"ohm tree $T$ that is productive, we show that $T$ can be stratified into distinct chunks of observations, and thereby $T$ is a canonical term.
    Starting with the root of $T$, the first chunk of observation will be along all the traces starting with the root of them and ending with a constant. 
    The traces must be finitary because of the productivity condition. The stratification can be repeated for each child term in the spines of the constants.
    
\end{proof}

\subsection{Hereditary Substitution}

The notion of hereditary substitution \cite{Watkins02tr,Harper07jfp} is 
used to define substitution on canonical terms in a typing-agnostic way.
The definition of hereditary substitution does not require the argument
terms to be well-typed in a dependently typed setting. 
In this way, we break the circular dependency between typing and substitution
in a dependent type theory.
The substitution is well-defined 
as long as a correct simple type of the argument term is provided.

The simple types $\tau$ are inductively 
defined by the following grammar.
\footnote{
    Using the syntax tree described in this paper, the grammar definition should be understood as $\tau\depth{k+1} ::= * \mid (\tau_1)\depth{k+1} \to (\tau_2)\depth{k+1}$.
    In subsequent discussions when we say $M\depth{k+1}$ has type $\tau_1 \to \tau_2$, it should be understood formally as 
    $M\depth{k+1}$ has type $(\tau_1)\depth{\omega} \to (\tau_2)\depth{\omega}$.
    For purely inductive definitions, the entire structure of the term can be revealed in a single observation,
    and
    we choose to omit the depth annotations completely for purely inductive definitions to reduce the visual clutter.
}
\[\tau ::= * \mid \tau_1 \to \tau_2\]

The hereditary substitution $[N\depth k/x]^\tau M\depth k$ is defined as along as 
$\Delta \vdash N\depth k : \tau$, where $\Delta = h_1 : \tau_1, \dots, h_n : \tau_n$ 
is a mapping from constants and variables to their simple types. Here, head $h$ refers to either $c$ or $x$.
The judgment is defined by induction on $k$ and the structure of $N$.

\boxed{\Delta \vdash N\depth k : \tau}
\begin{mathpar}
\inferrule{
}{
    \Delta \vdash  N\depth {0} : \tau
}

\inferrule{
    \Delta, x : \tau_1 \vdash N\depth {k+1} :  \tau_2
}{
    \Delta \vdash \lambda x. \, N\depth {k+1} : \tau_1 \to \tau_2
}

\inferrule{
    x : \tau_1 \to \dots \to \tau_n \to \tau \in \Delta \\
    \forall _{i, 1 \le i \le n}. \Delta \vdash (N_i)\depth {k+1} : \tau_i
}{
    \Delta \vdash x \cdot ((N_1)\depth {k+1}; \dots ; (N_n)\depth {k+1}) : \tau
}

\inferrule{
    c : \tau_1 \to \dots \to \tau_n \to \tau \in \Delta \\
    \forall _{i, 1 \le i \le n}. \Delta \vdash (N_i)\depth {k} : \tau_i
}{
    \Delta \vdash c \cdot ((N_1)\depth {k}; \dots ; (N_n)\depth {k}) : \tau
}
\end{mathpar}


We extend hereditary substitution
to infinite terms in the sense that given two infinite terms in their canonical form (i.e. $\beta$-normal-$\eta$-long form), 
there is a systematic procedure of generating an infinite term that is the result of 
substituting one term (for free variables) into another. In particular, if 
the two input terms to the hereditary substitution procedure are of an observation depth $k$, 
then the resulting term can be calculated up to the observation depth $k$.
The following judgments define hereditary 
substitution on productive B\"ohm trees. The type $\tau$ in the judgment provides typing information for the term being 
substituted by ($N\dk$ or $T\dk$). 

\begin{center}
\begin{longtable}{ll}
$[N\pdepth{k}/x]^\tau M\pdepth k  =\pdepth k M' \pdepth k$ & Hereditary substitution in canonical terms  \\
$[N\pdepth k/x]^\tau R\pdepth k  =\pdepth k R' \pdepth k$ & Hereditary substitution in neutral terms \\
$[N\pdepth k/x] ^\tau  T \pdepth k =\pdepth k T' \pdepth k$ & Hereditary substitution in continuing spines  \\
$T\pdepth k \rhd ^\tau N \pdepth k  = \pdepth k R' \pdepth k$ & Continuing spine applications \\
$[N\pdepth k/x] ^\tau  S \pdepth {k+1} =\dkpo S' \dkpo$ & Hereditary substitution in suspended spines  \\
\end{longtable}
\end{center}

One feature of hereditary substitution worth noting is that even if the type $\tau$ does not type the term 
being substituted for, the procedure still terminates but produces an undefined value, because no clauses will apply in the definition. In other words,
the procedure of substitution is robust with respect to typing of the input terms.
The typing information is to ensure that the procedure of hereditary substitution is well-defined as an inductive definition.

The judgments for hereditary substitution are defined by lexicographic induction on $\tau$, $k$ and the structure of the term 
on the right-hand side of $\tau$ ($M\dk$, $R\dk$, $T\dk$, $N\dk$, and $S\dkpo$ respectively) as follows.

    \begin{longtable}[l]{l}
\boxed{[N\pdepth{k}/x]^\tau M\pdepth k  =\pdepth k M' \pdepth k}\\
    $[N\pdepth{0}/x]^\tau M\pdepth{0} =\pdepth{0} M'\depth{0}$ \\
    $[N \pdepth{k+1}/x]^\tau R \pdepth{k+1} = \pdepth{k+1} [N \pdepth{k+1}/x]^\tau R \pdepth{k+1}$ \\
    $[N \pdepth{k+1}/x]^\tau \lambda y.\, M \pdepth{k+1} = \lambda y.  \,M'\dkpo$ \\ \hspace{12em} if $ [N \pdepth{k+1}/x]^\tau M \pdepth{k+1} = \pdepth{k+1} M'\pdepth{k+1}$ \\
\boxed{[N\pdepth k/x]^\tau R\pdepth k  =\pdepth k R' \pdepth k}\\
    $[N\pdepth{0}/x]^\tau R\pdepth{0} =\pdepth{0} R'\depth{0}$ \\
    $[N \pdepth{k+1}/x]^\tau (x \cdot T\pdepth{k+1} )=\pdepth{k+1} N' \pdepth{k+1} $ 
 \\ \hspace{12em} \makecell[lt]{if $[N\pdepth{k+1}/x]^\tau T \pdepth{k+1}=\pdepth{k+1} T'\pdepth{k+1}$
    \\ and $T' \pdepth{k+1} \rhd^\tau N \pdepth {k+1} =\pdepth{k+1} N'\pdepth{k+1}$}\\
    $[N \pdepth{k+1}/x]^\tau (y \cdot T\pdepth {k+1}) =\pdepth{k+1} y \cdot T' \pdepth{k+1}$ 
    \\ \hspace{12em} if $[N\pdepth{k+1}/x]^\tau T \pdepth{k+1}=\pdepth{k+1} T'\pdepth{k+1}$ \\
    $[N \pdepth{k+1}/x]^\tau (c \cdot S\pdepth {k+1}) =\pdepth{k+1} c \cdot S'  \pdepth{k+1}$ 
    \\ \hspace{12em} if $[N\pdepth{k}/x]^\tau S \dkpo=\pdepth{k+1} S'\dkpo$ \\
\boxed{[N\pdepth k/x] ^\tau  T \pdepth k =\pdepth k T' \pdepth k}\\ 
    $[N\pdepth{0}/x]^\tau T\pdepth{0} =\pdepth{0} T'\depth{0}$ \\
    $[N \pdepth{k+1}/x]^\tau \snil  =\pdepth{k+1} \snil$ \\
    $[N \pdepth{k+1}/x]^\tau (M\pdepth{k+1};T\pdepth{k+1}) =\pdepth{k+1} M'\pdepth{k+1};T'\pdepth{k+1} $ 
    \\ \hspace{12em}
    \makecell[lt]{if $[N\pdepth{k+1}/x]^\tau M\dkpo=\pdepth{k+1}M'\pdepth{k+1}$ 
    \\ and $[N\pdepth{k+1}/x]^\tau T\pdepth{k+1} =\pdepth{k+1}T'\pdepth{k+1}$} \\
\boxed{T\pdepth k \rhd ^\tau N \pdepth k  = \pdepth k R' \pdepth k}\\
    $T\pdepth{0} \rhd ^\tau N \pdepth 0 =\pdepth{0} R'\depth{0}$ \\
    $\snil  \rhd ^* R\pdepth{k+1} =\pdepth{k+1} R\pdepth{k+1}$ \\
    $(N\dkpo ; T\dkpo) \rhd^{\tau_2 \to \tau_1} \lambda x. M\dkpo =\pdepth{k+1} M'' \pdepth{k+1}$  
    \\ \hspace{12em} \makecell[lt]{if $[N\pdepth{k+1}/x]^{\tau_2} M\pdepth{k+1}  = \pdepth{k+1} M'\pdepth{k+1}$ 
    \\ and $T\pdepth{k+1} \rhd^{\tau_1}M'\pdepth{k+1}=\dkpo M''\pdepth{k+1}$}\\
\boxed{[N\pdepth k/x] ^\tau  S \dkpo =\dkpo S' \dkpo}\\ 
    $[N \pdepth{k}/x]^\tau \snil  =\pdepth{k+1} \snil$ \\
    $[N \pdepth{k}/x]^\tau (M\dk;S\dkpo) =\pdepth{k+1} M'\dk;S'\pdepth{k+1} $ 
    \\ \hspace{12em}
    \makecell[lt]{if $[N\pdepth{k}/x]^\tau M\dk=\dk M'\dk$ 
    \\ and $[N\pdepth{k}/x]^\tau S\pdepth{k+1} =\pdepth{k+1}S'\pdepth{k+1}$} \\
    \end{longtable}

     Note that in the substitution clauses for $[N\depth{k+1}/x](c \cdot S\dkpo)$, the premise assumes that the term $N\depth {k+1}$ is of depth $k$. 
    This is justified because, as we mentioned previously in Section~\ref{sec:equality_infinitary_syntax_tree}, any term may be viewed at a lesser observation depth.

\begin{theorem} [Hereditary Substitution Respects Observation Depth]
    If both $M\depth k$ and $N\depth k$ are terms of observation depth $k$, then for all $\tau$, $[N\depth k/x]^\tau M\depth k$  is of productive depth $k$ if defined. 
\end{theorem}
\begin{proof}
    Straightforward lexicographic induction on $\tau$, $k$, and the structure of $M\depth{k}$. 
\end{proof}




\begin{corollary}
    If $M$ (i.e., $M\depth \omega$) and $N$ (i.e. $N \depth \omega$) are infinitary productive terms, then so is $[N/x]^\tau M$ (i.e. $[N\depth \omega/x]^\tau M\depth{\omega}$). 
\end{corollary}
\begin{proof}
    The result of hereditary substitution can have arbitrary productive depth by the previous proof. 
\end{proof}

\begin{theorem}
    [Commutation of Hereditary Substitution]
    For all $k$, 

     if $[(N_1)\depth k/x]^{\tau_1}[(N_2)\depth k/z]^{\tau_2}M\depth k = \depth k M'\depth k$, 

    then $[[(N_1)\depth k/x]^{\tau_1}(N_2)\depth k/z]^{\tau_2}[(N_1)\depth k/x]^{\tau_1}M \depth k = \depth k M' \depth k$.
\end{theorem}
\begin{proof}
    By lexicographic induction on $k$, $\tau_2$, and the structure of $M$. 
\end{proof}


\section{The Type Theory of CoLF$^\omega$}
\label{sec:the_type_theory_of_colf_omega}
We present the type theory of a logical framework whose term model is typed productive B\"ohm trees.

\subsection{Syntax}
Besides the canonical terms, a logical framework also has the syntactic classes of canonical and atomic types, 
kinds, signatures, and contexts. We use ``expressions'' to refer to these entities (terms, types, kinds) in general. 
Expressions may contain potentially infinite terms with observation depths, and
so they are also parametrized by observation depths. 
Although the terms may be infinitary, the structure of kinds and types are finitary.
Thus, 
two expressions are equal up to depth $k$ if 
they are structurally equal and the underlying terms are equal up to depth $k$.
We write $A\depth k \to B\dk$ for $\Pi x : A\dk.\, B\dk$ if $x$ does not occur in $B\depth k$.
Similarly, we may write $A\dk \to K\dk$ for $\Pi x : A\dk. K\dk$. 
The syntax for signatures, contexts, kinds, and types are as follows. Notice that the depth remains $k+1$ on the 
right-hand side of all grammar rules for types and kinds, i.e. the type and kind structures are essentially inductive.
\begin{center}
    \begin{tabular}{lrcl}
        Signature & $\Sigma\depth{k+1}$ & $::= $ & $ \cdot \mid \Sigma, a : K\depth{k+1} \mid \Sigma, c : A\depth{k+1}$ \\
        Context & $\Gamma\depth{k+1}$ & $::= $ & $\cdot \mid \Gamma\depth{k+1}, x : A\depth{k+1}$ \\
        Kind & $K \depth{k+1}$ & $ ::= $ & $\type \mid \cotype \mid \Pi x : A\depth{k+1}.\ K\depth{k+1}$ \\
        Canonical types & $A\depth{k+1}, B\depth{k+1}$ & $ ::= $ & $P\depth{k+1} \mid \Pi x : A\depth{k+1}.\ B\depth{k+1}$ \\
        Atomic types & $P\depth{k+1}$ & $ ::= $ & $a \cdot S\depth{k+1}$ \\
        Canonical terms  &  $M\pdepth{k+1} $ & $ ::= $ & $\lambda x.\, M\pdepth{k+1} \mid R\pdepth{k+1}$ \\
        Neutral terms   &  $R\pdepth{k+1} $ & $ ::= $ & $x \cdot T\pdepth{k+1} \mid  c\cdot S \pdepth{k+1}$\\
        Continuing Spines   & $T\pdepth{k+1}$ & $ ::= $ & $\snil \mid M\pdepth {k+1} ; T\pdepth{k+1}$ \\
        Suspended Spines  & $S\pdepth{k+1}$ & $ ::= $ & $\snil \mid M\pdepth {k} ; S\pdepth{k+1}$ \\
    \end{tabular}
\end{center}

There is a  correspondence between the structure of canonical dependent types 
and the simple types. We use the erasure operation 
$(-)^o$ (super script $o$) to map an observable canonical type to a simple type $\tau$. For example, 
$(\Pi x:a.\,  a_2\cdot (x))^o =* \to *$. 
We also define the hereditary substitution of a canonical term into kinds, canonical and atomic types, 
and contexts. They are defined by the following judgments.

\begin{center}
    \begin{longtable}{llllll}
    $A\depth{k+1}^o = \tau$ & Observable type $A$ erases to $\tau$ \\
    $[N\depth k /x]^\tau K\dkpo =\dkpo K' \dkpo $ & Hereditary substitution in types \\
    $[N\depth k /x]^\tau A\dkpo  =\dkpo A' \dkpo $ & Hereditary substitution in canonical type \\
    $[N\depth k /x]^\tau P\dkpo  =\dkpo P' \dkpo $ & Hereditary substitution in atomic type \\
    $[N\depth k/x]^\tau \Gamma\dkpo  =\dkpo  \Gamma'\dkpo $ & Hereditary substitution in contexts \\
    \end{longtable}
\end{center}

The operation of substitution is also defined by lexicographic induction on $\tau$, $k$ 
and the structure of the expression on the right-hand side of $\tau$ ($K, A, P, \Gamma$). 
    \begin{longtable}[l]{llllll}
    \boxed{A\depth{k+1}^o = \tau}\\
    $((\pitype{x}{A_2}{A_1})\pdepth{k+1})^o = ((A_2)\pdepth{k+1}^o) \to ((A_1)\pdepth{k+1}^o)$ \\
    $(P\pdepth{k+1})^o = *$ \\
    \boxed{[N\depth k /x]^\tau K\dkpo  =\dkpo K' \dkpo } \\
    $[N\depth {k}/x]^\tau\type  =\depth{k+1} \type $ \\
    $[N\depth {k}/x]^\tau\cotype =\depth{k+1}\cotype$ \\
    $[N\depth {k}/x]^\tau\Pi y : A\dkpo.\ K\depth{k+1} = \Pi y: A'\dkpo. K'\dkpo$ 
    \\ \hspace{12em} \makecell[lt]{ if $ [N\depth{k}/x]^\tau A\depth{k+1} = \depth{k+1} A'\depth{k+1}$  
    \\ and $[N\depth{k}/x]^\tau K\depth{k+1} =\depth{k+1} K' \depth{k+1}$ }\\
    \boxed{[N\depth k /x]^\tau A\dkpo  =\dkpo A' \dkpo }\\
    $[N\dk /x]^\tau P\depth{k+1} = P'\dkpo $ 
    if $[N\dk /x]^\tau P\dkpo = \dkpo P' \dkpo$ \\
    $[N\dk /x]^\tau\Pi y : B\dkpo.\, A\dkpo =\depth{k+1}\Pi y : B'\dkpo.\, A'\depth{k+1}$ \\ 
    \hspace{12em} \makecell[lt]{
        if $[N\depth{k}/x]^\tau B\depth{k+1} = \depth{k+1} B'\depth{k+1}$ 
    \\ and $[N\depth{k}/x]^\tau A\depth{k+1} = \depth{k+1} A'\depth{k+1}$}\\
    \boxed{[N\depth k /x]^\tau P\dkpo  =\dkpo P' \dkpo }\\
    $[N\dk /x]^\tau(a \cdot S\depth{k+1}) =\depth{k+1} a \cdot S' \depth{k+1} $ 
    \\ \hspace{12em} if $[N\depth{k}/x]^\tau S\depth{k+1} = \depth{k+1} S'\depth{k+1}$ \\
    \boxed{[N\depth k/x]^\tau \Gamma\dkpo  =\dkpo  \Gamma'\dkpo } \\
    $[N\dk /x]^\tau \cdot =\depth{k+1} \cdot $ \\
    $[N\dk /x]^\tau(\Gamma\dkpo, y : A\depth{k+1}) =\depth{k+1} \Gamma'\dkpo, y : A\depth{k+1} $  
    \\ \hspace{12em} \makecell[lt]{if $ [N\depth{k}/x]^\tau\Gamma\depth{k+1} = \depth{k+1} \Gamma'\depth{k+1}$ 
    \\ and $[N/x]^\tau\depth{k} A\depth{k+1} = \depth{k+1} A'\depth{k+1}$ }\\
    \end{longtable}

\begin{theorem} [Hereditary Substitution Respects Observation Depth]
    If both $N\depth k$ is a term of observation depth $k$, 
    and $E\dkpo$ ($E$ is $K, A, P$) is a kind/type of depth $k + 1$,
    then for all $\tau$, $[N\dk/x]^\tau E\dkpo$  is of productive depth $k+1$ if defined. 
\end{theorem}
\begin{proof}
    Straightforward lexicographic induction on $\tau$, $k$, and the structure of $E\dkpo$. 
\end{proof}


\subsection{Type Checking Rules}

We simultaneously define the following type checking judgments, by induction on $k$ and the structure of the 
subject expression.
All judgments except $\Sigma\dk\validsig$ presuppose $\Sigma\dk \validsig$.
All judgments with $\Gamma\dk$ present presuppose $\vdash_{\Sigma\dk} \Gamma\dk \validctx$.

\begin{center}
    \begin{longtable}{ll}
         $\Sigma\depth k\validsig$ & Signature $\Sigma\dk$ is valid\\
         $\vdash_{\Sigma\depth k}\Gamma\depth k\validctx $ & Context $\Gamma\dk$ is well-formed\\
          $\Gamma\depth k\vdash_{\Sigma\depth k} K\dk \Leftarrow \kind$ & Kind $K\dk$ is a valid kind\\
          $\Gamma\depth k\vdash_{\Sigma\depth k} A\dk \Leftarrow \typecotype$ & Type $A\dk$ is a canonical (co)type\\
        $\Gamma\depth k\vdash_{\Sigma\depth k} P\dk \Rightarrow K\dk$ & Atomic type $P\dk$ synthesizes kind $K\dk$ \\
          $\Gamma\depth k\vdash_{\Sigma\depth k} S\dk \rhd K\dk \Rightarrow K'\dk$ & Suspended Spine $S\dk$ applied to kind $K\dk$ produces kind $K'\dk$\\
          $\Gamma\depth k\vdash_{\Sigma\depth k} M\dk \Leftarrow A\dk$ & Term $M\dk$ checks against type $A\dk$\\
          $\Gamma\depth k\vdash_{\Sigma\depth k} R\dk \Rightarrow P\dk$ & Neutral term $R\dk$ synthesizes type $P\dk$ \\
          $\Gamma\depth k\vdash_{\Sigma\depth k} T\dk \rhd A\dk \Rightarrow  P\dk$ & Continuing Spine $T\dk$ applied to type $A\dk$ produces type $P\dk$\\
          $\Gamma\depth k\vdash_{\Sigma\depth k} S\dk \rhd A\dk \Rightarrow  P\dk$ & Suspended Spine $S\dk$ applied to type $A\dk$ produces type $P\dk$\\
    \end{longtable}
\end{center}

$\boxed{\vdash\Sigma\dk \validsig}$ 
\begin{mathpar}
        \inferrule{
        }{
         \vdash    \Sigma\depth0 \validsig
        }

        \inferrule{
        }{
         \vdash    \cdot \validsig
        }
    
        \inferrule{
           \vdash \Sigma\dkpo \validsig
            \\ 
            \cdot\vdash K\dkpo \Leftarrow  \kind
        }{
           \vdash \Sigma\dkpo, a : K\dkpo \validsig
        }
    
        \inferrule{
            \vdash\Sigma\dkpo \validsig
            \\ 
            \cdot \vdash A\dkpo \Leftarrow  \typecotype{}
        }{
            \vdash\Sigma\dkpo, c : A\dkpo \validsig
        }
    \end{mathpar}

    $\boxed{\vdash_{\Sigma\dk}\Gamma\dk \validctx }$
\begin{mathpar}
        \inferrule{
        }{
            \vdash_{\Sigma\depth{0}}  \Gamma\depth{0} \validctx
        }

        \inferrule{
        }{
            \vdash_{\Sigma\dkpo} \cdot \validctx 
        }
    
        \inferrule{
            \vdash_{\Sigma\dkpo} \Gamma\dkpo \validctx 
            \\ 
            \Gamma\dkpo \vdash_{\Sigma\dkpo} A\dkpo \Leftarrow \typecotype
        }{
            \vdash \Gamma\dkpo, x : A\dkpo \validctx 
        }
    \end{mathpar}

    $\boxed{\Gamma\dk \vdash_{\Sigma\dk} K\dk \Leftarrow \kind}$
    \begin{mathpar}
        \inferrule{
        }{
            \Gamma\depth 0\vdash_{\Sigma\depth0}  K\depth {0} \Leftarrow  \kind
        }

        \inferrule{
        }{
            \Gamma\dkpo \vdash_{\Sigma\dkpo} \type \Leftarrow \kind
        }
    
        \inferrule{
        }{
            \Gamma\dkpo \vdash_{\Sigma\dkpo} \cotype \Leftarrow \kind
        }
    
        \inferrule{
            \Gamma\dkpo \vdash_{\Sigma\dkpo} A\dkpo \Leftarrow \typecotype
            \\
            \Gamma\dkpo, x : A\dkpo\vdash_{\Sigma\dkpo} K\dkpo \Leftarrow \kind
        }{
            \Gamma\dkpo \vdash_{\Sigma\dkpo}  \Pi x : A\dkpo . \, K\dkpo \Leftarrow \kind
        }
    \end{mathpar}
    
    $\boxed{\Gamma\dk \vdash_{\Sigma\dk} A\dk \Leftarrow \typecotype}$ 
    \begin{mathpar}
        \inferrule{
        }{
            \vdash_{\Sigma\depth{0}}  A\depth{0} \Leftarrow  \typecotype
        }

        \inferrule{
            \Gamma\dkpo \vdash_{\Sigma\dkpo} B\dkpo \Leftarrow \typecotype\\
            \Gamma\dkpo, x : B\dkpo\vdash_{\Sigma\dkpo} A\dkpo \Leftarrow \typecotype
        }{
            \Gamma\dkpo \vdash_{\Sigma\dkpo} \Pi x : B\dkpo. \, A\dkpo \Leftarrow \typecotype
        }
    
        \inferrule{
            \Gamma\dkpo \vdash_{\Sigma\dkpo} P\dkpo \Rightarrow K\dkpo \\
            K\dkpo = \type/\cotype
        }{
            \Gamma\dkpo \vdash_{\Sigma\dkpo} P\dkpo \Leftarrow \typecotype
        } 
    \end{mathpar}

    $\boxed{\Gamma\dk\vdash_{\Sigma\dk} P\dk \Rightarrow K\dk}$
    \begin{mathpar}
        \inferrule{
        }{
            \Gamma\depth0\vdash_{\Sigma\depth0} P\depth0 \Leftarrow K\depth {0} 
        }
    
        \inferrule{
            a : K\dkpo \in \Sigma\dkpo
            \\
            \Gamma\dkpo \vdash _{\Sigma\dkpo} S\dkpo \rhd K\dkpo \Rightarrow  K'\dkpo
        }{
            \Gamma\dkpo\vdash_{\Sigma\dkpo} a\cdot S\dkpo \Rightarrow  K'\dkpo
        }
    \end{mathpar}

    $\boxed{\Gamma\dk \vdash_{\Sigma\dk} S\dk \rhd K\dk \Rightarrow K'\dk} $
    \begin{mathpar}
        \inferrule{
        }{
            \Gamma\depth0\vdash_{\Sigma\depth0}  S\depth{0} \rhd K\depth{0} \Rightarrow K'\depth {0} 
        }

    \inferrule{
    }{
        \Gamma\dkpo  \vdash_{\Sigma\dkpo} \snil  \rhd K\dkpo \Rightarrow  K\dkpo
    }

    \inferrule{
        \Gamma\dkpo \vdash_{\Sigma\dkpo} M\dkpo \Leftarrow  A\dkpo \\
        [M\dk/x]^{\erased{(A\dkpo)}} K\dkpo =  \depth{k+1}K'\dkpo
        \\
        \Gamma\dkpo \vdash_{\Sigma\dkpo} S\dkpo \rhd K'\dkpo \Rightarrow  K''\dkpo
    }{
        \Gamma\dkpo  \vdash_{\Sigma\dkpo} M\dk ;S\dkpo  \rhd \Pi x : A\dkpo. \, K\dkpo \Rightarrow   K''\dkpo
    }
    \end{mathpar}

    \boxed{\Gamma\dk \vdash_{\Sigma\dk} M\dk \Leftarrow A\dk} 
    \begin{mathpar}
        \inferrule{
        }{
            \Gamma \depth0\vdash_{\Sigma\depth 0}  M\depth0 \Leftarrow A \depth {0} 
        }

    \inferrule{
        \Gamma\dkpo  \vdash_{\Sigma\dkpo} R\dkpo \Rightarrow P'\dkpo
        \\
        P'\dkpo =\depth{k+1} P\dkpo
    }{
        \Gamma\dkpo  \vdash_{\Sigma\dkpo} R\dkpo \Leftarrow P\dkpo
    }

    \inferrule{
        \Gamma\dkpo, x: B\dkpo \vdash_{\Sigma\dkpo} M\dkpo \Leftarrow A\dkpo
    }{
        \Gamma\dkpo  \vdash_{\Sigma\dkpo} \lambda x. M\dkpo \Leftarrow\Pi x:B\dkpo.\,A\dkpo
    }
    \end{mathpar}

    \boxed{\Gamma\dk \vdash_{\Sigma\dk} R\dk \Rightarrow P\dk} 
    \begin{mathpar}
        \inferrule{
        }{
            \Gamma\depth0\vdash_{\Sigma\depth0}  R\depth0 \Rightarrow P\depth0
        }

    \inferrule{
        x : A\dkpo \in \Gamma\dkpo
        \\
        \Gamma\dkpo \vdash _{\Sigma\dkpo} T\dkpo \rhd A\dkpo \Rightarrow  P\dkpo
    }{
        \Gamma\dkpo  \vdash_{\Sigma\dkpo} x  \cdot T\dkpo \Rightarrow  P\dkpo
    }

    \inferrule{
        c : A\dkpo \in \Sigma
        \\
        \Gamma\dkpo \vdash _{\Sigma\dkpo} S\dkpo \rhd A\dkpo \Rightarrow P\dkpo
    }{
        \Gamma\dkpo  \vdash_{\Sigma\dkpo} c \cdot S\dkpo \Rightarrow P\dkpo
    }

    \end{mathpar}

    \boxed{\Gamma\dk \vdash_{\Sigma\dk} T\dk \rhd A\dk \Rightarrow  P\dk} 
    \begin{mathpar}
        \inferrule{
        }{
            \Gamma\depth{0}\vdash_{\Sigma\depth{0}}  T\depth{0} \rhd A\depth{0} \Rightarrow P\depth{0}
        }

    \inferrule{
    }{
        \Gamma\dkpo  \vdash_{\Sigma\dkpo} \snil  \rhd P\dkpo \Rightarrow  P\dkpo
    }

    \inferrule{
        \Gamma\dkpo \vdash_{\Sigma\dkpo} M\dkpo \Leftarrow  B\dkpo \\
        [M\dk /x]^{\erased{B}} A\dkpo =  \depth{k+1}A' \dkpo
        \\
        \Gamma\dkpo \vdash_{\Sigma\dkpo} T\dkpo \rhd A'\dkpo \Rightarrow P\dkpo
    }{
        \Gamma\dkpo  \vdash_{\Sigma\dkpo} M\dkpo;T\dkpo  \rhd \Pi x : B\dkpo. \, A\dkpo \Rightarrow   P\dkpo
    }(*)

    \end{mathpar}

    \boxed{\Gamma\dk \vdash_{\Sigma\dk} S\dk \rhd A\dk \Rightarrow  P\dk} 
    \begin{mathpar}
        \inferrule{
        }{
            \vdash_{\Sigma\depth{0}}  S\depth{0} \rhd A\depth{0} \Rightarrow P\depth{0}
        }

    \inferrule{
    }{
        \Gamma\dkpo  \vdash_{\Sigma\dkpo} \snil  \rhd P\dkpo \Rightarrow  P\dkpo
    }

    \inferrule{
        \Gamma\dkpo \vdash_{\Sigma\dkpo} M\dkpo \Leftarrow  B\dkpo \\
        [M\dk /x]^{\erased{B}} A\dkpo =  \dkpo A' \dkpo
        \\
        \Gamma\dkpo \vdash_{\Sigma\dkpo} S\dkpo \rhd A'\dkpo \Rightarrow P\dkpo
    }{
        \Gamma\dkpo  \vdash_{\Sigma\dkpo} M\dk;S\dkpo  \rhd \Pi x : B\dkpo. \, A\dkpo \Rightarrow   P\dkpo
    }

    \end{mathpar}


We state and prove substitution lemmas for CoLF$^\omega$.
\begin{theorem}[Substitution]
    Given $\Sigma$, and depth $k$, we have
    \begin{enumerate}
        \item 
     Substitution in canonical terms: 

    If $\Gamma\dkpo \vdash N\dkpo \Leftarrow A\dkpo$, and $\Gamma\dkpo, x : A\dkpo, \Gamma'\dkpo\vdash M\dkpo \Leftarrow B\dkpo$, 

    then $\Gamma\dkpo, [N\depth k/x] \Gamma'\dkpo \vdash [N\dkpo/x] M\dkpo \Leftarrow  [N\dk/x] B\dkpo$.
    
        \item 
     Substitution in canonical types: 

    If $\Gamma\dkpo \vdash N\dkpo \Leftarrow \dkpo A$, and $\Gamma\dkpo, x : A\dkpo, \Gamma'\dkpo \vdash B\dkpo \Leftarrow \typecotype$, 
    then $\Gamma\dkpo, [N\depth k/x]\ \Gamma'\dkpo  \vdash [N\depth k/x] B\dkpo \Leftarrow \typecotype$.

    \end{enumerate}
\end{theorem}

\begin{proof}
    Induction on $\tau$, $k$, and the structure of $M$ or $B$. This proof is similar to the ones in \citet{Watkins02tr,Harper07jfp}.
\end{proof}

\subsection{Validity of Infinitary Terms}


With signature $\Sigma\depth\omega$, we say that a type $A$ is inductive if $A = \Pi x_1\dots \Pi x_n:A_n. a\cdot S$ and $a : \Pi y_1 \dots \Pi y_m : B_m.\, \type$, 
and a type $A$ coinductive if $A = \Pi x_1\dots \Pi x_n:A_n. a\cdot S$ and $a : \Pi y_1 \dots \Pi y_m : B_m.\, \cotype$.
A constructor $c$ is inductive if $c : A \in \Sigma$ and $A$ is inductive, and $c$ is coinductive if $c : A \in \Sigma$ and $A$ is coinductive.
As with CoLF \cite{Chen23fossacs}, a priority is assigned to each type: types and type families whose kinds are declared later in the signature 
have higher priority than types whose kinds are declared earlier. Term constructors inherit priorities from their types.

Intuitively, an infinitary term $M$ is valid if for all infinitary traces $T$ of $M$, some coinductive constructor 
occurs infinitely often along the trace $T$. In particular, 
a term in which there is an infinite stack of inductive constructors is not valid.

Recall that a trace is a sequence of constructor constants or variables, whose length is possibly infinite.
A term $M$ is \emph{trace-valid} if for all infinite traces $T$ of $M$, 
there is a coinductive constructor that occurs 
infinitely often along $T$, and that coinductive constructor has a higher priority 
than any other constructor that occurs infinitely often along $T$.
In other words, for any infinite trace $T$ of $M$, the highest priority constructor that occurs infinitely often 
along $T$ must be coinductive.

\begin{theorem}[Validity entails productivity]
    If $M$ is trace-valid, then $M$ is productive.
\end{theorem}
\begin{proof}
    Along any infinite trace $T$ through $M$, there exist infinitely many coinductive
     constructors along $T$ by trace-validity, and then it is productive by the presence of constructors.
\end{proof}

We note that hereditary substitution preserves the validity of terms.
\begin{theorem} [Trace-validity closed under hereditary substitution]
    If both $N$ and $M$ are trace-valid, then $[N/x]^{\tau} M$ is trace-valid if defined.
\end{theorem}
\begin{proof}
    An analysis of the definition of hereditary substitution reveals that 
    any infinite trace of $[N/x]^\tau M$ is a zipping of a trace $T_2$ of $N$ and $T_1$ of $M$, 
    at least one of which is infinite.
    Then since both $N$ and $M$ are productive, the highest priority constructor that occurs infinitely often 
    along $T_2$ or $T_1$ or both (depending on which ones are infinite) will be coinductive. Thus, 
    the infinite trace formed as a result of zipping will have an infinitely occurring coinductive constructor 
    of the highest priority.
\end{proof}

\section{Interpretation of finitary signatures}
\label{sec:interpretation_of_colf_signatures}

We now present a finitary signature $\Sigma$ in the style of CoLF \cite{Chen23fossacs} that can be interpreted into a CoLF$^\omega$ signature $\Sigma\depth\omega$.
Type checking in CoLF$^\omega$ is undecidable up to depth $\omega$, 
but is decidable for any finite depth $k < \omega$.
Moreover, if the recursive definitions in the finitary signature satisfy the prepattern restriction of CoLF, 
that all recursion constants are applied to bound variables, then we have decidable 
type checking up to depth $\omega$. 
The reason is that if the terms that appear in the signature are all rational, 
then the signature may be regarded as a CoLF signature for the purpose of type checking, while 
semantically remains to be interpreted in CoLF$^\omega$.

\begin{longtable}{lrcl}
    Signatures & $ \Sigma $ & $::= $ & $ \cdot \mid \Sigma, a : K \mid \Sigma, c : A \mid \Sigma, r : A = M  $ \\
    Contexts & $\Gamma$ & $::= $ & $\cdot \mid \Gamma, x : A $  \\
    Kinds & $K$ & $ ::= $ & $\type \mid \cotype \mid \Pi x : A.\ K $ \\
    Canonical types & $A, B$ & $ ::= $ & $P \mid \Pi x : B.\ A  $ \\
    Atomic types & $P$ & $ ::= $ & $a \cdot S$ \\
    Canonical terms & $M$ & $ ::= $ & $ R  \mid \lambda x.\, M $\\
    Neutral terms & $R$ & $ ::= $ & $ x \cdot S \mid c \cdot S \mid r \cdot S  $\\
    Spines & $S$ & $ ::= $ & $  \snil\mid M ; S  $\\
\end{longtable}
The infinitary terms are given by terms computed by a fixed point definition.
To ensure that recursive definitions are still contractive and productive, 
we require the head of $M$ in a definition $r : A = M$ to be a constant $c$.
\footnote{This turns out to be more restrictive than CoLF, because we would like to rule out definitions, e.g. $r = \lambda x. x\, (r\, x)$, which 
are not productive. CoLF signatures allow nonproductive terms because substitutions for recursive terms are restricted to variable renaming only. }
We define the expansion of a finitary CoLF signature $\Sigma$ into an infinitary CoLF$^\omega$ signature of depth $k$
by the function $\exp\dk$, and similarly for all syntax categories. The expansion functions are implicitly 
parametrized by a signature $\Sigma$, providing the recursive definitions.
We have two choices when expanding a CoLF spine $S$, 
a suspended spine $S\dk$ or a continuing spine $T\dk$. We use $\exp^S\dk$ for expansion into suspended spines
 and  $\exp^T\dk$ for expansion into continuing spines. The definition of the expansion is as follows.

 \begin{longtable} {lll}
    \boxed{\exp\dk(\Sigma)=\dk \Sigma\dk} \\
    $\exp\depth0(\Sigma) $ & $= \depth0\Sigma\depth0$\\
    $\exp\dkpo(\Sigma, a : K) $ & $=\dkpo \exp\dkpo(\Sigma), a : K\dkpo$\\
    $\exp\dkpo(\Sigma, c : A) $ & $=\dkpo \exp\dkpo(\Sigma), c : A\dkpo$\\
    $\exp\dkpo(\Sigma, r : A = M) $ &$ =\dkpo \exp\dkpo(\Sigma)$\\
    \boxed{\exp\dk(\Gamma)=\dk \Gamma\dk} \\
    $\exp\depth0(\Gamma) = \depth0\Gamma\depth0$\\
    $\exp\dkpo(\Gamma, x : A) $ & $=\dkpo \exp\dkpo(\Gamma), x : A\dkpo$\\
    \boxed{\exp\dk(K)=\dk K\dk} \\
    $\exp\depth0(K) = \depth0K\depth0$\\
    $\exp\dkpo(\type) $ & $=\dkpo \type$\\
    $\exp\dkpo(\cotype) $ & $=\dkpo \cotype$\\
    $\exp\dkpo(\Pi x : A.\, K) $ & $=\dkpo \Pi x : \exp\dkpo(A).\, \exp\dkpo(K)$\\
    \boxed{\exp\dk(A)=\dk A\dk} \\
    $\exp\depth0(A) = \depth0A\depth0$\\
    $\exp\dkpo(P) $ & $=\dkpo \exp\dkpo(P)$\\
    $\exp\dkpo(\Pi x : B.\, A) $ & $=\dkpo \Pi x : \exp\dkpo(B).\, \exp\dkpo(A)$\\
    \boxed{\exp\dk(P)=\dk P\dk} \\
    $\exp\depth0(P) = \depth0P\depth0$\\
    $\exp\dkpo(a \cdot S) $ & $=\dkpo a \cdot \exp^S\dkpo(S)$\\
    \boxed{\exp\pdepth k(M)=\pdepth{k}M\dk} \\
    $\exp\pdepth {0}(
        M
    ) $ & $ = \pdepth 0 M\depth0$ \\
    $\exp\pdepth {k+1}(
        \lambda x.\, M
    ) $ & $ = \pdepth{k+1} \lambda x. \exp\pdepth{k+1}(M)$ \\
    $\exp\pdepth {k+1}(
        R
    ) $ & $ = \pdepth{k+1} \exp\dkpo(R)$ \\
    \boxed{\exp\pdepth k(R)=\pdepth{k}R\dk}
    \\
    $\exp\pdepth {0}(
        R
    ) $ & $ = \pdepth 0 R\depth0$ \\
    $\exp\pdepth {k+1}(
        c \cdot S
    ) $ & $ =\pdepth{k+1} c\cdot (\exp^S\pdepth{k}(S))$\\
    $\exp\pdepth {k+1}(
        x \cdot S
    ) $ & $ = \pdepth{k+1}x\cdot (\exp^T\pdepth{k}(S))$\\
    $\exp\pdepth {k+1}(
       r \cdot S
    ) $ & $ = \pdepth{k+1} \exp^T\pdepth{k+1}(S) \rhd^{\tau} \exp\dkpo(M)$ 
    \\& if $r : A = M \in \Sigma$ and $\tau = (\exp\dkpo(A))^o$\\
    \boxed{\exp^T\pdepth k(S)=\pdepth{k}T\dk}
    \\
    $\exp^T\pdepth {0}(
        S
    ) $ & $ = \pdepth 0 T\depth0$ \\
    $\exp^T\pdepth {k+1}
        \snil
     $ & $ =\pdepth{k+1} \snil$\\
    $\exp^T\pdepth {k+1}(
        M ; S
    ) $ & $ = \pdepth{k+1}(\exp\pdepth{k+1}(M)); (\exp^T\pdepth{k+1}(S))$\\
    \boxed{\exp^S\pdepth k(S)=\pdepth{k}S\dk}
    \\
    $\exp^S\pdepth {0}(
        S
    ) $ & $ = \depth0 S \depth0$ \\
    $\exp^S\pdepth {k+1}
        \snil
     $ & $ =\pdepth{k+1} \snil$\\
    $\exp^S\pdepth {k+1}(
        M ; S
    ) $ & $ = \pdepth{k+1}(\exp\pdepth{k}(M)); (\exp^S\pdepth{k+1}(S))$\\
    \end{longtable}

\begin{theorem}
    [Semi Decidability of Type Checking]
    Given a finitary signature $\Sigma$, it is decidable whether $\vdash\dk \exp\dk(\Sigma) \validsig$ for any $k$.
    And similarly for other type checking judgments.
\end{theorem}
\begin{proof}
    Given a finitary signature, its depth $k$ expansion can be computed. The type checking judgments are defined by induction on
    the observation depth and are thus decidable.
\end{proof}

\subsection{Infinitary  Canonical Forms}

Given a finitary signature $\Sigma$, 
let $\Sigma\depth\omega$ be the expansion $\m{exp}\depth\omega(\Sigma)$.
For a type $a : \type$ or $a : \cotype$ in $\Sigma$, 
the possibly infinitary canonical terms $M\depth \omega$ for the type $a$ in a context $\Gamma\depth\omega$ are such that 
$\Gamma\depth\omega\vdash_{\Sigma\depth\omega}  M\depth\omega \Leftarrow a$, and $M\depth\omega$ is valid.
$\Gamma\depth\omega\vdash_{\Sigma\depth\omega}  M\depth\omega \Leftarrow a$ means that for all observation depth $k$, the partial observation of $M\depth\omega$ at depth $k$, $M\dk$, satisfies
$\Gamma\dk \vdash_{\Sigma\dk} M \dk \Leftarrow a$. 

The above definition may generalize to type families indexed by terms. That is,
for a type family 
$a : \Pi x_1: A_1.\,\dots, \Pi x_n : A_n.\,\type$ or 
$a : \Pi x_1: A_1.\,\dots, \Pi x_n : A_n.\,\cotype$,
the possibly infinitary canonical terms $M\depth \omega$ for the type $a$ indexed 
by a spine of possibly infinitary terms, $a\cdot((M_1)\depth\omega, \dots, (M_n)\depth\omega)$ in a context $\Gamma\depth\omega$
are such that 
$\Gamma\depth\omega\vdash_{\Sigma\depth\omega}  M\depth\omega \Leftarrow a\cdot ((M_1)\depth\omega, \dots, (M_n)\depth\omega)$, 
and $M\depth\omega$ is valid.
That is, for all observation depth $k>0$, the partial observation of $M\depth\omega$ at depth $k$, $M\dk$, satisfies
$\Gamma\dk \vdash_{\Sigma\dk} M \dk \Leftarrow a\cdot ((M_1)\depth{k-1}, \dots, (M_n)\depth{k-1})$, and $M\depth\omega$ needs 
to satisfy the global validity condition.
In both cases, we may omit the phrase ``in a context $\Gamma\depth\omega$'' and just say $M\depth\omega$ is a canonical term for a type $a$ with a possible spine, if $\Gamma\depth\omega$ is an empty context.

The reader may verify the signatures presented in Section~\ref{sec:examples_of_colf_omega} 
provide adequate encodings of infinitary objects. As with Section~\ref{sec:examples_of_colf_omega}, 
we use the Twelf's concrete syntax \cite{Pfenning99cade} in the presentation of finitary signatures, where 
$\Pi$-abstractions are written using curly braces and $\lambda$-abstractions are written using square brackets. 
Implicit $\Pi$-abstractions are created for capitalized types.

\section{Encoding Productive B\"ohm Trees}
\label{sec:encoding_productive_bohm_trees}

We present an encoding of $\bot$-free B\"ohm trees  (not necessarily productive) in 
CoLF$^\omega$. Then we present two ways to encode productivity.

Below is the signature ($\Sigma_4$) for encoding $\bot$-free B\"ohm trees. \verb#ctm# stands for canonical terms,  
\verb#ntm# stands for neutral terms, \verb#vars# stands for a distinguished
set of variables, and \verb#consts# stands for a distinguished set of constants. Both \verb#ctm# and \verb#spine# are finite because we don't want 
terms consisting of infinitely many consecutive $\lambda$-abstractions or spines of infinitary length.
\label{colf:signature_four}
\begin{multicols}{2}
\begin{verbatim}
%% Signature 4: 
ctm : type.
spine : type.
ntm : cotype.
vars : type.
consts : type.
lam : (vars -> ctm) -> ctm.
base : ntm -> ctm.
varntm : vars -> spine -> ntm.
constntm : consts -> spine -> ntm.
snil : spine.
scons : ctm -> spine -> spine.
\end{verbatim}
\end{multicols}

The term $I = \lambda x. \, x$ may be encoded as \verb#lam ([x] base (varntm x snil))#. 
The term $Z = \lambda x. \, x \cdot (\lambda y. \, y \cdot (\lambda z. \, z \cdot (\dots)))$, not productive and not typeable, 
can be encoded as the $\Sigma_4$-typeable term \verb#lam ([x] base (varntm x (scons (...) snil)))# of type \verb#ctm#. 
The \verb#...# will be another encoding of the term $Z$, continuing indefinitely. 
Let the encoding of $Z$ subsequently be referred to as \verb#tmZ#.
\verb#tmZ# is a valid term because of the coinductive constructor \verb#varntm#
infinitely occurring along the only one infinite trace through \verb#tmZ#. 
Note that \verb#tmZ# also has infinitely many finite traces such as \verb#lam, base, varntm, scons, snil#, which 
validity does not concern.

\begin{theorem}
    [Adequacy of Encoding]
    There is a bijection between the $\bot$-free B\"ohm trees and 
    the infinitary canonical terms of type \verb$ctm$.
\end{theorem}
\begin{proof}
    We impose an observation structure on $\bot$-free B\"ohm trees, such that a single observation 
    may reveal  a $\lambda$-abstraction, a head-spine form, or the next element in a spine. Then, we can 
    establish bijective correspondences between $\bot$-free B\"ohm trees and their encodings by induction on the 
    observation depth.

    The validity condition on the canonical terms of type \verb$ctm$ 
    corresponds to the exact requirement of $\bot$-free B\"ohm trees, that 
    there are no infinite chains of $\lambda$-abstractions or applications.
\end{proof}

\subsection{Encoding the Productivity Condition}

We present two different ways we can incorporate productivity conditions, an 
external approach and an internal approach. 
Both approaches encode object-level productivity as framework-level validity.
The external approach encodes the productivity 
condition externally as a predicate on terms, and no modification to the previous signature is required.
We define new predicates beginning with \verb#prod# on all three syntactic classes and clauses for them below ($\Sigma_5$). 
The suffix \verb#c# means canonical terms, \verb#n# means neutral terms, \verb#t# means 
continuing spines, and \verb#s# means suspended spine. Constants are followed by suspended spines, which are the only 
progress points.
\begin{verbatim}
%% Signature 5:
prodc : ctm -> type.
prodn : ntm -> type.
prodt : spine -> type.
prods : spine -> cotype.


p1 : ({x} prodc (M x)) -> prodc (lam M).
p2 : prodn R -> prodc (base R).
p3 : prodt S -> prodn (varntm V S).
p4 : prods S -> prodn (constntm C S).
p5 : prodt snil.
p6 : prodc M -> prodt S -> prodt (scons M S).
p7 : prods snil.
p8 : prodc M -> prods S -> prods (scons M S).
\end{verbatim}

The rules are defined to be syntax-directed so given any term \verb#t# of type \verb#ctm#, 
a search for the derivation \verb#prodc t# (i.e. a term of type \verb#prodc t#) will deterministically pick rules 
between \verb#p1# and \verb#p8# corresponding to the syntactic structure. 
For example, let the term \verb#tmI# be the encoding of $I = \lambda x.\, x$, 
a search for a term of the type \verb#prodc tmI# will be \verb#p1 ([x] p2 (p3 p5))#, which is a valid 
term because it is finite. 

Within the signature $\Sigma_5$, observe that the only coinductive type is \verb#prods#, and the 
only coinductive constructor constructors are \verb#p7# and \verb#p8# because they are constructors for the coinductive type family $\mt{prods}$.
Given a nonproductive term \verb#t#, 
there cannot be a valid term of type \verb#prodc t#.
Every infinite trace of \verb#t# will correspond to some infinite trace in a term of type \verb#prodc t#, 
and for every constant along that trace, each element in its spine will correspond to a coinductive constructor \verb#p8#. 
Since \verb#t# is not productive, 
there exists an infinite trace with only finitely many constants, and then the 
corresponding infinite trace of \verb#prodc t# will contain only finitely many \verb#p8#'s, 
and thus is not a valid term. For example, consider the supposed 
construction of the term \verb#pTmZ# of type  \verb#prodc tmZ# where \verb#tmZ# encodes $Z$, 
since there are no constants in $Z$, the term \verb#pTmZ# does not involve \verb#p7# or \verb#p8#.
It is infinite because it follows the syntax of \verb#tmZ#, and therefore it cannot be valid.
In summary, the predicate beginning with \verb#prod# encodes the productivity of the underlying 
term.

\begin{theorem}
    [Adequacy of Encoding]
    Given a $\bot$-free B\"ohm tree $M$, and its infinite encoding $\enc{M}$, 
    the $M$ is valid if and only if the type $\mt{prodc} \,\enc{M}$ is inhabited.
\end{theorem}
\begin{proof}
    Since \verb$prodc$ is defined in a syntax-directed way,  
    given a term $M$, there exists a unique term $M'\depth\omega = \enc{M}$ such that 
    $\vdash_{\Sigma\depth\omega} M'\depth\omega \Leftarrow \mt{prodc}\,M\depth\omega$.
    Then, it suffices to show that the validity condition on $M'\depth\omega$ encodes the productivity condition.
    Any infinite trace of $M'\depth\omega$ corresponds to an infinite trace of $M$, and the productivity condition on 
    the infinite traces of $M$ exactly correspond to the validity condition of $M'\depth\omega$.

\end{proof}

Let us now turn to the internal approach. In the internal approach, instead of having 
a separate predicate encoding the validity, we postulate that only productive infinite terms 
are well-typed terms, and nonproductive are invalid by definition. 
The signature ($\Sigma_6$) below shows the internal approach.
The progress points are modeled through progress canonical terms ($\mt{p\_ctm}$), 
which appear only in suspended spines ($\mt{s\_spine}$) following constants.

\begin{minipage}{\textwidth}
\begin{multicols}{2}
\begin{verbatim}
%% Signature 6:
ctm : type.
ntm : type.
t_spine : type.
s_spine : cotype.
p_ctm: cotype.
vars : type.
consts : type.

lam : (vars -> ctm) -> ctm.
base : ntm -> ctm.
varntm : vars -> t_spine -> ntm.
constntm : consts -> s_spine -> ntm.
tnil : t_spine.
tcons : ctm -> t_spine -> t_spine.
snil : s_spine.
scons : p_ctm -> s_spine -> s_spine.
progress : ctm -> p_ctm.
\end{verbatim}
\end{multicols}
\end{minipage}

In $\Sigma_6$, the only coinductive constructor is \verb#progress#, which is used 
to guard the terms in suspended spines following constants.
We can see that the productivity of infinitary $\lambda$-terms is 
encoded as the validity of CoLF$^\omega$. 
The encoding of $I = \lambda x.\, x$ will be \verb#lam ([x] base (varntm x tnil))#, which is valid because it is finite.
The encoding of $Z = \lambda x. \, x \cdot (\lambda y. \, y \cdot (\lambda z. \, z \cdot (\dots)))$ is 
\verb#lam ([x] base (varntm x# \verb#(tcons (...) tnil)))#, where \verb#...# denotes the same 
encoding of $Z$. This term is not valid because the encoding does not contain the only 
coinductive constructor \verb#progress# and there exists an infinite trace through the term.
If we fix $c$ as a constant, $\mt{c : consts}$,
the term $Y = \lambda x.\, c \cdot (x \cdot (c \cdot ( x \cdot (\dots))))$, is encoded as
\verb$lam ([x] base (constntm c (scons (progress (base (varntm x (tcons (...)$ \verb$tnil)))) snil)))$.
This encoding is valid because of the occurrences of $\mt{progress}$ along the infinite trace through the term

\begin{theorem}
    [Adequacy of Encoding]
    There is a bijection between the productive $\bot$-free B\"ohm trees $\lambda$-terms and 
    the infinitary canonical terms of type \verb$ctm$.
\end{theorem}
\begin{proof}
    As with the previous proof, we stratify the $\bot$-free B\"ohm trees $\lambda$-terms by observation depths,
     and establish the bijective correspondence with infinitary canonical terms of type \verb$ctm$.
    The validity condition is established similarly by a condition on traces: for each trace in a productive B\"ohm tree, there is a corresponding trace in the infinitary canonical term of type \verb$ctm$.
    The productivity corresponds to the validity: each occurrence of a constant implies the occurrences of the 
    \verb$progress$ coinductive constructor for each argument on its spine.
\end{proof}

\section{Case Study: Co-Natural Numbers and Co-Binary Numbers}
\label{sec:conatural_numbers_and_cobinary_numbers}

Recall the definition of unary conatural numbers \cite{Chen23fossacs}, which are 
exactly natural numbers together with $\infty$ where $\infty = \m{cosucc}\, \infty$.

\begin{verbatim}
conat : cotype.
cozero : conat.
cosucc : conat -> conat.
\end{verbatim}

Similar to the way we define binary numbers for ``little-endian'' presentations \cite{Pfenning19talkplmw}, 
where we observe the least significant digits first in an inductive datatype, we define 
the analogous ``little-endian'' presentation for conats using binary streams. 
In this representation, every conatural number that is finite has a unique representation: \verb$b1$ and \verb$b0$'s 
followed by an infinite stack of \verb$b0$'s. The conatural number $\infty$ has infinitely 
many representations, as long as \verb$b1$'s occur infinitely often in the stream.
\footnote{The reader may notice that the cobinary numbers are defined exactly the same as the real numbers in Section~\ref{sec:real_numbers}.
Thus, the definition of the $\mt{bplus}$ relation later is also a definition of the sum of two real numbers.}

\begin{verbatim}
cobin : cotype.
b0 : cobin -> cobin.
b1 : cobin -> cobin.
\end{verbatim}
\begin{theorem}
   [Adequacy of Encoding]

   There exists a bijection between the cobinary numbers and the canonical terms of type \verb$cobin$.
\end{theorem}

\begin{proof}
   We again characterize cobinary numbers by observation, and prove the correspondence directly by induction on the observation depth and the structure of the terms or
   the structure of cobinary numbers.

\end{proof}

We have the following examples of cobinary numbers. \verb$bzero$ is an encoding of the cobinary $0$.
\verb$bone$ is an encoding of the cobinary number $1$. \verb$w1$ and \verb$w2$ are two examples of the cobinary number $\infty$.
\verb$bsucc$ is an encoding of the ``successor'' relation between two cobinary numbers. 
 
\begin{verbatim}
bzero : cobin = b0 (bzero).
bone : cobin = b1 (bzero).
w1 : cobin = b1 (w1).
w2 : cobin = b1 (b0 w2).

bsucc : cobin -> cobin -> cotype.
bsucc/0 : bsucc (b0 X) (b1 X).
bsucc/1 : bsucc X Y -> bsucc (b1 X) (b0 Y).
\end{verbatim}
\begin{theorem}
   [Adequacy of Encoding]
   \verb$bsucc B1 B2$ is inhabited 
   if and only if \verb$B2$ is a successor of \verb$B1$. 
\end{theorem}
\begin{proof}
    Directly by induction on the observation depth, and the structure of the term 
    or the structure of the informal proof that \verb$B2$ is a successor of \verb$B1$.
\end{proof}

We may define conversions \verb$tobin$ and \verb$frombin$ establishes 
between unary numbers and binary numbers.

\begin{verbatim}
coplus : conat -> conat -> conat -> cotype.
coplus/0 : coplus cozero Y Y.
coplus/1 : coplus X Y Z
                -> coplus (cosuss X) Y (cosucc Z).

tobin : conat -> cobin -> cotype.
tobin/0 : tobin cozero bzero.
tobin/1 : 
    tobin X Y
    -> bsucc Y Z
    -> tobin (cosucc X)  Z.

frombin : cobin -> conat -> cotype.
frombin/0 : frombin X Y
    -> coplus Y Y Z
    -> frombin (b0 X) Z.
frombin/1 : frombin X Y
    -> coplus Y Y Z
    -> frombin (b1 X) (cosucc Z).
\end{verbatim}

We then define the cobinary plus relation. The cobinary plus sums up two cobinary numbers. 
The definition will not require a base case as cobinary numbers are defined coinductively.
\begin{verbatim}
bplus : cobin -> cobin -> cobin -> cotype.
bplus/00 : 
    bplus X Y Z ->
    bplus (b0 X) (b0 Y) (b0 (Z)).
bplus/01 : 
    bplus X Y Z ->
    bplus (b0 X) (b1 Y) (b1 (Z)).
bplus/10 : 
    bplus X Y Z ->
    bplus (b1 X) (b0 Y) (b1 (Z)).
bplus/11 : 
    bplus X Y Z ->
    bsucc Z W ->
    bplus (b1 X) (b1 Y) (b0 W).
\end{verbatim}

We show some simple invariants of binary streams, such as that $0+0=0$, and that the successor of $0$ is $1$.

\begin{verbatim}
b0+0is0 : bplus bzero bzero bzero =  bplus/00 (b0+0is0).
bsucc0is1 : bsucc bzero bone = bsucc/0.
\end{verbatim}
 
 

We may show that \verb$bplus$ is indeed sound by 
presenting an encoding of the proof for the following \emph{cobinary sum soundness} theorem:
For all cobinary numbers $N$ and $M$, if $M$ is the ``successor'' of $N$, then 
the cosuccessor of the unary conatural number corresponding to $N$ is 
equal to the cosuccessor of the unary conatural number corresponding to $M$.
The statement of the theorem is encoded as the \verb$bsucc_sound$ type family and 
the proof is encoded as terms inhabiting the type family. 
This technique follows the Twelf \cite{Pfenning99cade,Schurmann03tphol}.
 

\begin{verbatim}
eqconat : conat -> conat -> cotype.
eqconat/0 : eqconat cozero cozero.
eqconat/1 : eqconat X Y -> eqconat (cosucc X) (cosucc Y).

eqconat/refl : eqconat X X -> cotype.
eqconat/refl/0 : eqconat/refl eqconat/0.
eqconat/refl/1 : eqconat/refl EQR
    -> eqconat/refl (eqconat/1 EQR).

bsucc_sound : {N}{M} {CN} {CM} frombin N CN -> frombin M CM 
            -> bsucc N M -> eqconat (cosucc CN) CM -> cotype.
bsucc_sound/0 : 
    eqconat/refl EQC
    -> bsucc_sound (b0 N') (b1 N')  
        (frombin/0 N'FB N'CP) 
        (frombin/1 N'FB N'CP) 
        bsucc/0 
        (eqconat/1 EQC).
bsucc_sound/1 : 
    bsucc_sound N' M' N'FB M'FB SN'isM' EQC
    ->
    bsucc_sound (b1 N') (b0 M')  
        (frombin/1 N'FB N'CP) 
        (frombin/0 M'FB M'CP) 
        (bsucc/1 SN'isM') 
        (eqconat/1 EQC).
\end{verbatim}

\begin{theorem}
   [Adequacy of Encoding]
   There is a bijection between the informal derivation showing the cobinary sum soundness theorem 
   and the canonical terms of the fully instantiated type family type \verb$bsucc_sound$.

   That is, there is a bijection between the canonical terms of the 
    type \[\mt{bsucc\_sound}\, \enc{N}\, \enc{M}\,
   \enc{\mi{CN}}\,\enc{\mi{CM}}\,\enc{\mi{SNM}}\,\enc{\mi{EQSCNCM}}\] and the derivation (proof) of the following 
   theorem: for all cobinary numbers $N$ and $M$, if $M$ is the ``successor'' of $N$ (evidenced 
by a derivation $\mi{SNM}$), then 
$\mi{CN}$, the cosuccessor of the unary conatural number corresponding to $N$, is 
equal to $\mi{CM}$, the cosuccessor of the unary conatural number corresponding to $M$, 
evidenced by $\mi{EQSCNCM}$.
\end{theorem}
\begin{proof}
    The derivation for the proof of the theorem is infinite.
    Then the adequacy proof follows directly by induction on the 
    observation depth, and then the structure of  term or the proof.
\end{proof}

\section{Related Work}

\textbf{Observation Depth.} The idea of observation depth is 
inspired by prior works in handling infinitary structures, such as non-terminating computations in 
the semantics of programming languages, including
step-indexed logical relations \cite{Ahmed04phd}, 
the semantics of mixed induction and coinduction in subtyping systems for recursive types \cite{Lakhani22esop} and 
in sized types \cite{Abel16jfp,Somayyajula22fscd},
the tree topology of B\"ohm trees \cite{Barendregt85book} and infinitary lambda calculus \cite{Kennaway97,Severi17mscs}, 
and prior works in the semantics and metatheory of CoLF \cite{Chen21ms,Chen23fossacs}.

\textbf{Logical Frameworks.} The use of typed $\lambda$-terms to model syntax trees dates back to 
the logical framework LF \cite{Harper93jacm}. The idea of using only $\beta$-normal-$\eta$-long of typed $\lambda$-terms is first 
seen in the development of Concurrent LF \cite{Watkins02tr}, along with the 
technique of hereditary substitution. The metatheory of LF with only canonical terms 
is further developed in the research of 
encoding the simply typed $\lambda$-calculus and its metatheorems in Twelf \cite{Pfenning99cade,Watkins02tr,Harper07jfp}. 
Twelf also incorporates a logic programming engine on LF signatures 
and mechanical checking for metatheorems \cite{Pfenning91book,Schurmann03tphol}. 
The development of a logic programming interpretation of CoLF$^\omega$ will be future work.

\textbf{Infinitary $\lambda$-Calculus.}
B\"ohm trees were first devised as a method for studying the solvability of $\lambda$-terms \cite{Barendregt85book}, and serve to 
semantically distinguish between $\lambda$-terms that do not have a normal form.
The $\bot$-free B\"ohm trees are also called hereditary head normal terms \cite{Tatsuta08flops}.
Various other forms of infinitary $\lambda$-calculus were studied as representations of programs involving recursion \cite{Huet98mscs,Kennaway97,Severi17mscs}.
However, most studies on infinitary $\lambda$-calculus are done in an untyped setting.
The notion of  productivity comes from Coquand's work on defining functions on infinitary objects 
in a type theory \cite{Coquand93lncs}, and from cut-elimination in linear logic with fixed points \cite{Fortier13csl,Derakhshan19arxiv,Derakhshan20arxiv}.
The notion of validity follows from a similar definition of CoLF \cite{Chen23fossacs} and the Horn $\mu$-calculus \cite{Charatonik98lics}.

\textbf{Levels of Priorities.}
In CoLF$^\omega$,
simple proofs by induction and proofs by coinduction have trivial embedding in this framework.
Almost all practical mixed inductive and coinductive structures and mixed inductive and coinductive proofs 
including subtyping systems \cite{Lakhani22esop}, use a simple two-level of priorities, where the induction is 
nested inside the coinduction. 
Our definition of infinitary syntax trees follows a similar approach where 
the observation depth serves as the coinductive progress point.
The validity condition, however, assumes 
an infinite number of levels of priorities, following the prior work of CoLF \cite{Chen23fossacs} 
and proof systems with mixed induction and coinduction \cite{Fortier13csl,Charatonik98lics}.

\section{Conclusion}

We have presented an interpretation of CoLF-style signatures that the term model may be arbitrary 
finitary and infinitary terms. We take the notion of finitary observation as central
to the definition of infinitary structures and characterize the 
infinite syntax trees inductively in terms of observation depth.  The equality on infinitary terms 
is just a bisimulation on the infinitary terms. We then characterize 
productivity as a subset of the infinitary syntax trees that are closed under hereditary substitution.
Typing structures are subsequently imposed on productive B\"ohm trees to define
 a dependently typed logical framework CoLF$^\omega$.
Finally, we have applied CoLF$^\omega$ in the encoding of productive B\"ohm trees and cobinary numbers
as case studies.

\textbf{Acknowledgement.} We thank Frank Pfenning for insightful discussions and comments on the drafts of this paper.
We also thank Robert Harper for his comments on an early draft of this paper.

\bibliographystyle{plainnat}
\bibliography{citationsconf}

\end{document}